\documentclass[11pt,letterpaper,onecolumn]{article}
\usepackage[margin=1in]{geometry}

\usepackage{graphicx}
\usepackage{multirow}
\usepackage{varwidth}
\usepackage{verbatim}
\usepackage{tabularx}
\usepackage{hyperref}
\usepackage{listings}
\usepackage{xcolor}
\usepackage[T1]{fontenc}
\usepackage{textcomp}
\usepackage{amsmath}
\usepackage{amssymb}
\usepackage{epsfig}
\usepackage{amsthm}
\usepackage{algorithm}
\usepackage{algpseudocode}
\usepackage{fancyvrb}
\usepackage{amssymb}
\usepackage{subcaption}
\usepackage{ dsfont }
\algrenewcommand\textproc{\emph}
\usepackage{ amssymb }

\algblock{ParFor}{EndParFor}
\algnewcommand\algorithmicparfor{\textbf{parallel for}}
\algnewcommand\algorithmicpardo{\textbf{do}}
\algnewcommand\algorithmicendparfor{\textbf{end\ parallel for}}
\algrenewtext{ParFor}[1]{\algorithmicparfor\ #1\ \algorithmicpardo}
\algrenewtext{EndParFor}{\algorithmicendparfor}

\let\oldReturn\Return
\renewcommand{\Return}{\State\oldReturn}

\newtheorem{theorem}{\emph{Theorem}}
\newtheorem{lemma}{\emph{Lemma}}

\theoremstyle{condition}

\theoremstyle{definition}
\newtheorem{definition}{\emph{Definition}}

\theoremstyle{corollary}

\theoremstyle{proposition}

\theoremstyle{remark}

\title{Auditable Register Emulations}
\author{Vinicius V. Cogo \and Alysson Bessani}
\date{LASIGE, Faculdade de Ci\^encias, Universidade de Lisboa, Lisboa, Portugal}

\begin{document}

\maketitle

\begin{abstract}
The widespread prevalence of data breaches amplifies the importance of auditing storage systems.
In this work, we initiate the study of auditable storage emulations, which provide the capability for an auditor to report the previously executed reads in a register.
We precisely define the notion of auditable register and its properties, and establish tight bounds and impossibility results for auditable storage emulations in the presence of faulty storage objects.
Our formulation considers loggable read-write registers that securely store data using information dispersal and support fast reads.
In such a scenario, given a maximum number~$f$ of faulty storage objects
and a minimum number~$\tau$ of data blocks required to recover a stored value, we prove that (1)~auditability is impossible if $\tau \leq 2f $; (2)~implementing a weak form of auditability requires $\tau \geq 3f+1$; and (3)~a stronger form of auditability is impossible.
We also show that signing read requests overcomes the lower bound of weak auditability, while totally ordering operations or using non-fast reads enables strong auditability.
\end{abstract}

\vspace{5mm}
\section{Introduction}
\label{sec:introduction}

Characteristics like cost-effectiveness, high scalability, and ease of use promoted the migration from private storage infrastructures to public multi-tenant clouds in the last decade.
Security and privacy concerns were the main deterrents for this migration since the beginning.
Numerous secure storage systems have been proposing the use of advanced cryptographic primitives to securely disperse data across multiple clouds (i.e., independent administrative domains~\cite{racs}) to reduce the risk of data breaches (e.g., DepSky~\cite{DepSky} and AONT-RS~\cite{resch2011aontrs}).

Given a secure storage system composed of $n$~storage objects, these \emph{information dispersal} techniques split and convert the original data item into $n$~coded blocks~\cite{Krawczyk1993secret,plank2013erasure,Rab89}.
Each coded block is stored in a different base object and clients need to obtain only~$\tau$ out of $n$~coded blocks to effectively recover the original data.
In this type of solution, no base object stores the whole data item, which differentiates information dispersal from fully-replicated storage systems (e.g.,~\cite{abraham2004byzantine,attiya1995sharing,malkhi1998byzantine})---where each object retains a full copy of the data.

Despite the advances in secure storage, the increasing severity of data breaches and the tightening of privacy-related regulations (e.g., GDPR) have been driving the demand for further improvements on this topic.
For instance, \emph{auditability}~\cite{ko2011towards} enables the systematic verification of who has effectively read data in secure storage systems.
Notably, this verification allows one to separate these users from the whole set of clients that are authorised to read data but have never done so.
It is an important step to detect data breaches (including those caused by authorised users, e.g., Snowden's case), analyse leakages, and sanction misuses.

\paragraph*{Problem.}
\label{subsec:problem}

In this paper, we address the following question: \emph{How to extend resilient storage emulations with the capability of auditing who has effectively read data?}
More specifically, we address the problem of protecting storage systems from readers trying to obtain data without being detected (i.e., audit completeness) and protecting correct readers from faulty storage objects trying to incriminate them (i.e., audit accuracy).
The answer must encompass the techniques used in these emulations, such as information dispersal~\cite{Krawczyk1993secret,Rab89} and available quorum systems~\cite{malkhi1998byzantine}, for providing a R/W register abstraction~\cite{lamport1986interprocess}.
We consider information dispersal the primary form of auditable register emulations because alternative solutions that replicate the whole data (even if encrypted) can suffer from faulty base objects leaking data (and encryption keys) to readers without logging this operation.

\paragraph*{Related Work.}
\label{subsec:related}

Several auditing schemes were proposed to verify the integrity of data stored in multi-tenant external infrastructures~\cite{kolhar2017cloud,popa2011enabling}.
They mainly focus on cryptographic techniques to produce retrievability proofs without the need to fetch all data from the system (e.g.,~\cite{juels2007pors}) or on providing public integrity verification (e.g.,~\cite{wang2013privacy}).
However, to the best of our knowledge, \emph{there is no previous work on auditing who has effectively read data in a dispersed storage}.

Another topic related to our work is the accountability, which focuses on making systems' components accountable in a way their actions become non-repudiable~\cite{haeberlen2010case}.
Works in the accountability literature have discussed generic scenarios for networked systems~\cite{Haeberlen2007PeerReview}, described the necessary cryptographic building blocks~\cite{haeberlen2010case}, or how evidences should be stored and protected~\cite{Haeberlen2007PeerReview}.

Several other works have explored the space complexity of fault-tolerant register emulations (e.g.,~\cite{aguilera2003using,cadambe2016information,chockler2007amnesic,chockler2017space}), including disintegrated storage~\cite{berger2018integrated} (e.g.,~\cite{androulaki2014erasure,DepSky,spiegelman2016space}).
However, \emph{none of these works focuses on the requirements for auditing read accesses in a storage system despite the existence of faulty storage objects}.

\paragraph*{Contributions.}
\label{subsec:contribution}
This paper initiates the study of auditing in resilient storage by presenting lower bounds and impossibility results related with the implementation of an \emph{auditable register} on top of $n$~base objects that log read attempts despite the existence of $f$~faulty ones.
Our results show that, given a minimum number~$\tau$ of data blocks required to recover a data item from information dispersal schemes and a maximum number~$f$ of faulty storage objects
(1)~auditability is impossible with $\tau \leq 2f$; (2)~when fast reads (reads executed in a single communication round-trip~\cite{Guerraoui06fastreads}) are supported, $\tau \geq 3f+1$ is required for implementing a weak form of auditability, while a stronger form of auditability is impossible; (3)~signing read requests overcomes the lower bound of weak auditability; and (4)~totally ordering operations or using non-fast reads can still provide strong auditability. 

\section{Preliminaries}
\label{sec:preliminaries}

\paragraph*{System Model.}
\label{subsec:system}

Our system is composed of an arbitrary number of \emph{client processes} $\Pi = \{p_1, p_ 2, ...\}$, which interact with a set of $n$~\emph{storage objects} $\mathds{O}=\{o_1,o_2,...,o_n\}$.
Clients can be subdivided into three main classes: \emph{writers} $\Pi_W = \{p_{w1}, p_{w2}, ...\}$, \emph{readers} $\Pi_R = \{p_{r1}, p_{r2}, ...\}$, and \emph{auditors} $\Pi_A = \{p_{a1}, p_{a2}, ...\}$.
These different roles do not necessarily mean they have to be performed by different processes.

A \emph{configuration} $C$ is a vector of the states of all entities (i.e., processes and objects) in the system.
An \emph{initial configuration} is a specific configuration where all entities of the system are in their initial state.
An \emph{algorithm} $\mathit{Alg}$ defines the behaviour of processes in $\Pi$ and objects in~$\mathds{O}$ as deterministic state machines, which can modify the system's states through \emph{actions} (e.g., invoke and response).
Clients \emph{invoke} operations in the storage objects and wait for \emph{responses}.
An \emph{execution segment} is a (finite or infinite) sequence of alternated configurations and actions.
An \emph{execution} $\rho$ is an execution segment that begins in an initial configuration.
An \emph{event} is the occurrence of an action in an execution.

A finite sequence of invocations and responses compose a \emph{history}~$\sigma$ of the system.
We consider the following relationships between operations within a history.
First, if a history~$\sigma$ contains the invocation of an operation~$op_1$ and its response, then $op_1$ is \emph{complete} in~$\sigma$.
Otherwise, if~$\sigma$ contains only the former, then operation $op_1$ is \emph{incomplete}.
Second, if the response of an operation $op_1$ precedes the invocation of another operation $op_2$ in~$\sigma$, then $op_1$ \emph{precedes} $op_2$.
Third, if operations $op_1$ and $op_2$ do not precede each other in~$\sigma$, then they are \emph{concurrent}.
Fourth, a history~$\sigma$ that does not contain concurrent operations is \emph{sequential}.

\paragraph*{Information Dispersal.}
We consider a high-level shared storage object that stores a value~$v$ from domain $\mathds{V}$ using information dispersal schemes (i.e., a special case of disintegrated storage~\cite{berger2018integrated}).
These schemes provide efficient redundancy since they disperse pieces of information across multiple locations instead of replicating the whole value on each one of them.
Moreover, they allow clients to recover the original stored value only after obtaining a subset of these pieces of information.
Examples of information dispersal schemes include erasure codes~\cite{plank2013erasure,Rab89} and secret sharing~\cite{Krawczyk1993secret,shamir79secretsharing}.

This object provides two high-level operations: \emph{a-write($v$)} and \emph{a-read()}.
A high-level \emph{a-write($v$)} operation converts a value $v \in \mathds{V}$, passed as an argument, into $n$~coded blocks $b_{v_1},b_{v_2},...,b_{v_n}$ from domain $\mathds{B}$, and each coded block~$b_{v_k}$ is stored in the $k^{th}$ base object in $\mathds{O}$ (as described bellow).
The employed techniques guarantee that no untrusted base object~$o_k$ stores an entire copy of value~$v$ and no reader recovers $v$ by obtaining less than a certain fraction of these blocks.

A high-level \emph{a-read()} operation recovers the original value~$v$ from any subset of a specific number~($\tau$) of blocks $b_{v_k}$.
It means readers do not need to execute low-level reads in all $n$~base objects to obtain these~$\tau$ blocks.
Base objects in this work are \emph{loggable R/W registers}.

\paragraph*{Loggable R/W Register Specification.}
\label{subsec:register}

A \emph{loggable read-write (R/W) register} is an object $o_k$ that stores a data block~$b_{v_k}$ from domain $\mathds{B}$ and has a log $L_k$ to store records of every read operation that this base object has responded.
This object $o_k$ provides three low-level operations~\cite{lamport1986interprocess}: 
\begin{itemize}
\item \emph{rw-write($b_{v_k}$)}: writes the data block~$b_{v_k} \in \mathds{B}$, passed as an argument, in this base object $o_k$ and returns an \emph{ack} to confirm that the operation has succeeded.
\item \emph{rw-read()}: returns the data block~$b_{v_k} \in \mathds{B}$ currently stored in this base object $o_k$ (or $\perp \notin \mathds{B}$ if no block has been written on it).
\item \emph{rw-getLog()}: returns the log $L_k$ of this base object $o_k$.
\end{itemize}

The behaviour of a loggable R/W register is given by its sequential specification.
A low-level \emph{rw-write} operation always writes the block passed as an argument and returns an \emph{ack}.
A low-level \emph{rw-read} operation returns the block written by the last preceding \emph{rw-write} operation on this object or $\perp \notin \mathds{B}$ if no write operation has been executed in this object.
Additionally, the low-level \emph{rw-read} operation also creates a record $\langle p_r, \mathit{label}(b_{v_k}) \rangle$ about this read operation in its log $L_k$, where $p_r$ is the identifier of the reader that invoked the \emph{rw-read} operation and \emph{label}($b_{v_k}$) is an auxiliary function that, given a block~$b_{v_k}$, returns a label associated with the value~$v$ (from which the block~$b_{v_k}$ was derived). 
A low-level \emph{rw-getLog} returns the log $L_k$ containing records of all preceding \emph{rw-read} executed in this object $o_k$.
A history~$\sigma$ of low-level operations in this object is linearisable~\cite{herlihy1990linearizability} if invocations and responses of concurrent operations can be reordered to form a sequential history that is correct according to its sequential specification.
This object is linearisable if all of its possible histories are also linearisable.

\paragraph*{}
We define a \emph{providing set} based on the notion introduced by Lamport~\cite{Lamport2006accepting} to abstract the access to multiple base objects.
A providing set is the set of base objects that have both stored a block associated with value~$v$ and responded this block to a reader~$p_r$ in history~$\sigma$.

\begin{definition}[Providing Set]{The set of objects $P^{\sigma}_{p_r,v}$ is providing blocks associated to value~$v$ for the reader $p_r$ in a history~$\sigma$ iff~$\sigma$ contains, for every object $o_k \in P^{\sigma}_{p_r,v}$, an event in which $o_k$ receives a write request to store $b_{v_k}$ and an event in which $o_k$ responds $b_{v_k}$ to a read request from $p_r$.}
\label{def:providing_set}
\end{definition}

Given a finite history~$\sigma$ being audited, there is only one possible providing set $P^{\sigma}_{p_r,v}$ for each pair $\langle p_r, v \rangle$ of reader $p_r$ and value~$v$.
For the sake of simplicity, this set is presented only as $P_{p_r,v}$ from now on.
Every single base object $o_k$ that stored block~$b_{v_k}$ and returned this block to $p_r$ belongs to the providing set $P_{p_r,v}$, independent on how many times $o_k$ has returned $b_{v_k}$ to $p_r$.
There might exist cases where providing sets for different pairs of readers and values are composed of the same objects.

We introduce the notion of an \emph{effective read}, which characterises a providing set large enough that a reader $p_r$ is able to effectively read value~$v$ from the received blocks.

\begin{definition}[Effective read]{
A value $v \in \mathds{V}$ is \emph{effectively read} by a reader $p_r$ in history~$\sigma$ iff $|P^{\sigma}_{p_r, v}| \geq \tau$, i.e., $p_r$ has received at least~$\tau$ distinct blocks~$b_{v_k}$ associated with value~$v$ from different base objects $o_k \in P^{\sigma}_{p_r, v}$ in history~$\sigma$.
}
\label{def:effective_read}
\end{definition}

An effective read depends only on the number of different coded blocks (derived from the same value~$v$) that a reader $p_r$ has already obtained from different base objects in the same history.
Reader~$p_r$ does not necessarily obtain all these blocks in a single high-level \emph{a-read} operation.
Due to unlimited concurrency, there might exist cases where it is accomplished only after receiving responses from many subsequent \emph{a-read} operations.

High-level operations in our system have two additional particularities.
First, high-level read operations (i.e., \emph{a-read}) are \emph{fast reads}.

\begin{definition}[Fast read~\cite{Guerraoui06fastreads}]{An \emph{a-read} operation is \emph{fast} if it completes in a single communication round-trip between the reader and the storage objects.}
\label{def:fast_read}
\end{definition}

Second, we consider multi-writer multi-reader (MWMR) shared objects where multiple writers from $\Pi_w$ execute high-level \emph{a-write} operations and multiple readers from $\Pi_ R$ execute high-level \emph{a-read} operations.

\paragraph*{Fault Model.}
\label{subsec:faults}

Clients (i.e., writers, readers, and auditors) and storage objects that obey their specifications are said to be \emph{correct}.
\emph{Faulty writers} and \emph{faulty auditors} are honest and can only fail by crashing.
\emph{Faulty readers} may crash or request data to only a subset of objects, which may characterise an attack where it attempts to read data without being detected and reported by the audit.

Writers are trusted entities because they are the data owners, who are the most interested party in the auditable register we are proposing in this work.
This assumption is common (e.g.,~\cite{DepSky,malkhi1998byzantine}) since malicious writers can always write invalid values, compromising the application's state.
Furthermore, auditors are also trusted because they are controlled either by the same entity as writers or by third-party entities writers trust.

\emph{Faulty storage objects} can be nonresponsive~(NR)~\cite{Jayanti1998nonresponsive} in an execution since they may crash, omit their values to readers, omit read records to auditors, or record nonexistent read operations. 
More specifically, they can fail by NR-omission when accessing their values and NR-arbitrary when accessing their log records~\cite{Jayanti1998nonresponsive}.
Omitting records to auditors means a faulty object may be helping a reader to evict being detected by the auditor.
Producing records for nonexistent reads characterises an attack where a faulty object may be trying to incriminate a reader.
Furthermore, we assume no more than $f$~storage objects are faulty.

\section{Auditable Register Emulations}
\label{subsec:extended}

We extend the aforementioned register emulation with a high-level operation \emph{a-audit()}.
It uses the fail-prone logs obtained from the \emph{rw-getLog} operation in loggable R/W registers to define an \emph{auditable register emulation}.
This emulation has access to a \emph{virtual log} $L \subseteq \bigcup _{k \in \{1..n\}} L_k$ from which we can infer who has effectively read a value from the register.

The \emph{a-audit} operation obtains records from $L$ and produces a \emph{set of evidences} $E_A$ about the effectively read values.
We define a threshold $t$ as the \emph{required number of collected records} to create an evidence $\mathcal{E}_{p_r,v}$ of an effective read.
Each \emph{evidence} $\mathcal{E}_{p_r,v}$ contains at least $t$~records from different storage objects~$o_k$ proving that $v$ was \emph{effectively read} by reader~$p_r$ in history~$\sigma$.
This threshold~$t$ is a configurable parameter that depends on the guarantees \emph{a-audit} operations provide (defined bellow).

We define an \emph{auditing quorum} $A$ as the set of objects from which the \emph{a-audit} collects fail-prone individual logs to compose the set of evidences.
To ensure audit operations are available in our asynchronous system model, we consider $|A| = n -f$.
We are interested in auditing effective reads because we aim to audit who has actually read a data value~$v$---including faulty readers that do not follow the read protocol and leave operations incomplete.
A correct auditor receives $E_A$ and reports all evidenced reads.

An \emph{auditable register} provides an \emph{a-audit()} operation that guarantees \emph{completeness} (i.e., Definition~\ref{def:completeness}) and at least one form of \emph{accuracy} (i.e., Definitions~\ref{def:weakacc} and \ref{def:strongacc}).

\begin{definition}[Completeness]
\label{def:completeness}
Every value~$v$ \emph{effectively read} by a reader $p_r$ before the invocation of an \emph{a-audit} in history~$\sigma$ is reported in the $E_A$ from this audit operation, i.e., $\forall p_{r} \in \Pi_R, \forall v \in \mathds{V}, |P_{p_r,v}|\geq \tau \implies \mathcal{E}_{p_r,v} \in E_A$.

\end{definition}

\begin{definition}[Weak Accuracy]
\label{def:weakacc}
A correct reader $p_r$ that has never invoked an \emph{a-read} (i.e., never tried to read any value~$v$) before the invocation of an \emph{a-audit} in history~$\sigma$ will not be reported in the $E_A$ from this audit operation, i.e., $\forall p_{r} \in \Pi_R, \forall v \in \mathds{V}, P_{p_r, v} = \varnothing \implies \mathcal{E}_{p_r,v} \notin E_A$.
\end{definition}

\begin{definition}[Strong Accuracy]
\label{def:strongacc}
A correct reader $p_r$ that has never effectively read a value~$v$ before the invocation of an \emph{a-audit} in history~$\sigma$ will not be reported in the $E_A$ from this audit operation as having read $v$, i.e., $\forall p_{r} \in \Pi_R, \forall v \in \mathds{V}, |P_{p_r, v}| < \tau \implies \mathcal{E}_{p_r,v} \notin E_A$.
\end{definition}

While completeness is intended to protect the storage system from readers trying to obtain data without being detected, accuracy focuses on protecting \emph{correct readers} from faulty storage objects incriminating them.

Both variants of accuracy guarantee that auditors report only readers that have invoked read operations.
The difference between them is that strong accuracy guarantees that auditors report only readers that have effectively read some value, while weak accuracy may report readers that have not effectively read any value.
Strong accuracy implies weak accuracy.
Additionally, the accuracy property provides guarantees to correct readers only, which means that auditors may report faulty (honest or malicious) readers in incomplete or partial reads because they had the intention to read the data and aborted it or crashed while doing so.

We consider \emph{weak auditability} when the storage system provides completeness and weak accuracy in audit operations and \emph{strong auditability} when it provides completeness and strong accuracy.

\section{Preliminary Results}
\label{subsec:mechanisms}

In the presence of faulty base objects, it is impossible to audit system in which readers can recover values by accessing a single object (e.g., \emph{non-replicated} and \emph{fully-replicated} systems).
The reason is that, in these solutions, every base object stores a full copy of data and can give it to readers without returning the log record of this operation to auditors.
Therefore, we consider \emph{information dispersal} techniques (e.g.,~\cite{Krawczyk1993secret,Rab89}) as the primary form of auditable register emulations.
Moreover, we assume these information dispersal schemes require $\tau > f$ blocks to recover the data because, otherwise, the $f$~faulty objects may also deliver their blocks to readers and omit the records of these operations from auditors.

\paragraph*{Records Available for Auditing Registers.}

As a starting point, we identify the minimum number of records from each preceding effective read that will be available for any \emph{a-audit} operation.
It relates to the minimum number of correct objects that are present at the smallest intersection of any providing set $P_{p_r,v}$ and an auditing quorum~$A$.

\begin{lemma}{}
\label{lem:min-records}
Any available a-audit operation obtains at least $\tau -2f$ records of every preceding effective read in history~$\sigma$.
\end{lemma}
\begin{proof}
Let us assume a system configuration composed of $n$~storage objects divided into four groups~$G_{1-4}$, as depicted in Figure~\ref{fig:scenario}.
Objects within the same group initially contain the same value~$v$ and an empty log $L_k = \varnothing$.
Without loss of generality, group~$G_1$ contains the $f$~faulty objects of the system, $G_2$ contains $\tau -2f$ correct objects, $G_3$ contains $f$~correct objects, and $G_4$ contains the remaining $n-\tau$ correct objects of the system.

\begin{figure}[h!]
\centering
\includegraphics[width=0.7\textwidth]{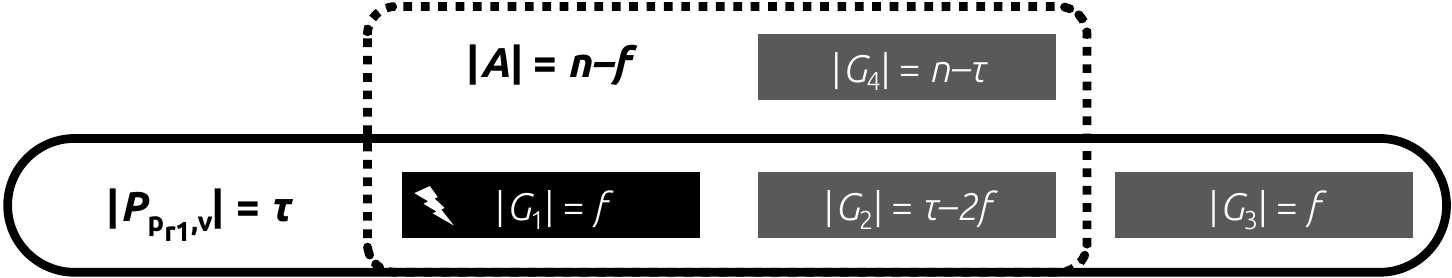}
\caption{\label{fig:scenario}A configuration with a providing set $P_{p_{r1},v}$ of an effective read and an auditing quorum~$A$.}
\end{figure}

A reader $p_{r1}$ effectively reads a value~$v$ after obtaining exactly~$\tau$ blocks $b_{v_k}$ from objects in a providing set $P_{p_{r1},v}$ composed of the $f$~faulty objects in $G_1$ (that do not record the read operation) and $\tau - f$ correct objects in $G_2 \cup G_3$ (that do record it).
In the worst-case scenario, an available \emph{a-audit} operation receives $n-f$ \emph{rw-getLog} responses from an auditing quorum~$A$ composed of the $f$~faulty objects in $G_1$ (that belong to the providing set $P_{p_{r1},v}$ but return empty logs $L_k$ to auditors), $\tau -2f$ correct objects in $G_2$ (that belong to the providing set $P_{p_{r1},v}$ and return the correct records of this read operation), and $n - \tau$ correct objects in $G_4$ (that do not belong to $P_{p_{r1},v}$).

In this case, auditors receive only the $\tau -2f$ records from the objects in $G_2$, which are the only correct objects at the intersection of $P_{p_{r1},v}$ and $A$.
This is the minimum number of records of an effective read available to an auditing quorum since any other configuration makes auditors receive more records of this effective read, i.e., in any alternative scenario, the faulty objects in $G_1$ may correctly record the read operation and return these records to auditors or the auditing quorum~$A$ may also include up to $f$~correct objects from $G_3$ (that belong to $P_{p_{r1},v}$) instead of objects from $G_4$.
\end{proof}

Receiving this minimum number $\tau -2f$ of records from an effective read must be enough to trigger the creation of an evidence for that read in \emph{a-audit} operations.
As a consequence, the value of $t$ must be defined considering this constraint.

\section{Impossibility Results}
\label{sec:resilience}

This section identifies impossibility results for the properties of auditable registers: completeness (Definition~\ref{def:completeness}), weak accuracy (Definition~\ref{def:weakacc}), and strong accuracy (Definition~\ref{def:strongacc}).

\begin{lemma}{}
\label{lem:completeness}
It is impossible to satisfy the \textbf{completeness} of auditable registers with $\tau \leq 2f$.
\end{lemma}

\begin{proof}
Consider an auditable register implemented using the same four subset groups of objects $G_{1-4}$ from Lemma~\ref{lem:min-records} (depicted in Figure~\ref{fig:scenario}).
Without loss of generality, let us assume that $\tau=2f$.
As a consequence, group $G_2$ will be empty because $|G_2| = \tau -2f = 0$.

A reader $p_{r1}$ obtains~$\tau$ coded blocks for value~$v$ from a providing set $P_{p_{r1},v}$ composed of $f$~faulty objects in $G_1$ (which do not log the operation) and $\tau-f = f$ correct nodes in $G_3$ (that log the operation).
This reader can decode the original value~$v$ after receiving these~$\tau$ blocks, performing thus an effective read.
Consequently, $|G_2| = 0$ and any auditing quorum~$A$ that does not include at least one object from $G_3$ will not receive any record for the read operation from reader $p_{r1}$.
For instance, an auditing quorum $A = G_1 \cup G_2 \cup G_4$ (e.g., from Figure~\ref{fig:scenario}) receives no record of this read in this history, which violates the completeness property since an evidence cannot be created without records (i.e., $t \geq 1$).
\end{proof}

\begin{lemma}{It is impossible to satisfy the \textbf{weak accuracy} of auditable registers with $t \leq f$.}
\label{lem:weak}
\end{lemma}
\begin{proof}

Consider an auditable register implemented using $n-f$ correct objects and $f$~faulty objects.
A read has never been invoked in this history of the system, but the $f$~faulty objects create records for a nonexistent read operation from a reader $p_{r}$.
Any auditing quorum~$A$ that includes these faulty objects will receive $f$~records for a read operation that has never been invoked.
If the threshold $t \leq f$ is enough to create an evidence and report an effective read, then a correct auditor must report it, violating the weak accuracy property.
\end{proof}

The next theorem shows in which conditions these two properties cannot be supported simultaneously.

\begin{theorem}{It is impossible to satisfy both \textbf{completeness} and \textbf{weak accuracy} of auditable registers with $\tau \leq 3f$.}
\label{the:weakcompl}
\end{theorem}
\begin{proof}

Consider an auditable register implemented using the same four subset groups of objects $G_{1-4}$ from Lemma~\ref{lem:min-records} (depicted in Figure~\ref{fig:scenario}).
Without loss of generality, let us assume that $\tau = 3f$ is enough to satisfy both the completeness and weak accuracy properties.
As a consequence, group $G_2$ will contain $|G_2| = \tau -2f = f$ objects.

A reader $p_{r1}$ obtains~$\tau$ blocks for value~$v$ from a providing set $P_{p_{r1},v}$ composed of $f$~blocks from the faulty objects in $G_1$ (which do not log the operation) and other $\tau -f$ correct objects in $G_2 \cup G_3$ (that log the operation).
The auditing quorum~$A$ will produce $\tau -2f = f$ records (Lemma~\ref{lem:min-records}).
Evidences must be created using these $f$~records to report the effective read and satisfy completeness.
However, as proved in Lemma~\ref{lem:weak}, it is impossible to guarantee weak accuracy with $t \leq f$.
\end{proof}

Now, we turn our attention to strong accuracy, i.e., the capability of an auditor to report exactly which value~$v$ each reader $p_r$ has effectively read.

\begin{lemma}{It is impossible to satisfy the \textbf{strong accuracy} of auditable registers with $t < \tau + f$.}
\label{lem:strong}
\end{lemma}
\begin{proof}

Consider an auditable register implemented using five subset groups of objects $G_{1-5}$, as depicted in Figure~\ref{fig:scenario2}.
The main difference to the scenario from Figure~\ref{fig:scenario} is that the group~$G_4$ was subdivided into two groups: $G_4$ with $2f-1$ objects and $G_5$ with $n-\tau-2f+1$ objects.

\begin{figure}[h!]
\centering
\includegraphics[width=0.7\textwidth]{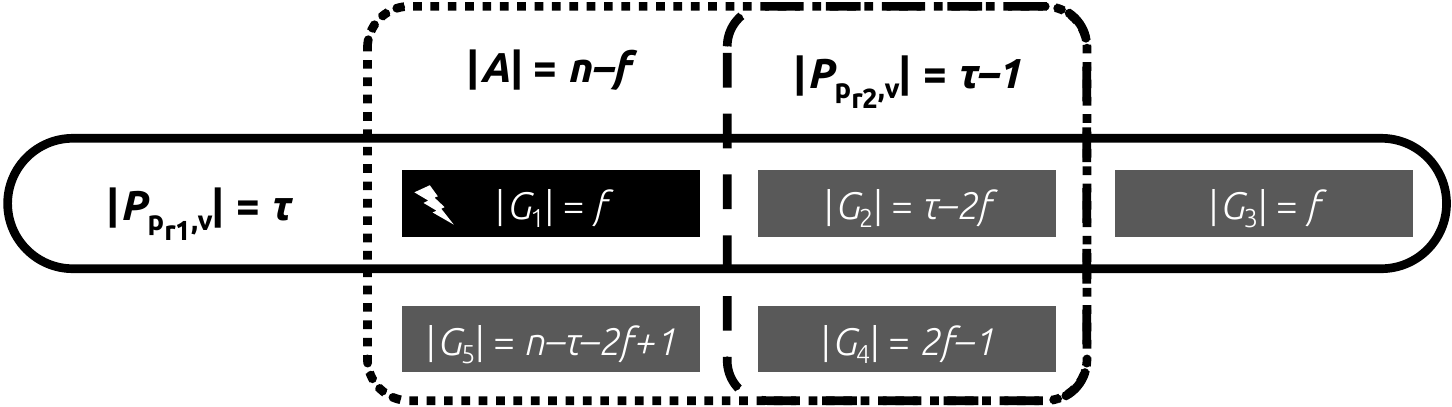}
\caption{\label{fig:scenario2}A configuration with two providing sets and an auditing quorum~$A$. Only $P_{p_{r1},v}$ represents an effective read.}
\end{figure}

This configuration has an incomplete write operation\footnote{An incomplete write operation may be caused, for instance, by network delays in clients' requests or a writer that has crashed in the middle of a write operation.} of a value $x \in \mathds{V}$ that has arrived to objects from all groups but $G_2$ and $G_4$.
A reader $p_{r2}$ obtains blocks of value~$v$ from objects from groups $G_2$ and $G_4$, which log the record for the read, and form a providing set with size $|P_{p_{r2},v}| = \tau - 2f +2f -1 = \tau -1$.
In this history, due to the incomplete write operation of value $x$, reader $p_{r2}$ has not received~$\tau$ blocks for value~$v$, which means it cannot recover~$v$.

However, the $f$~faulty objects in $G_1$ decide to mimic groups $G_2$ and $G_4$ and log the record for the read of value~$v$ by reader $p_{r2}$.
Any auditing quorum~$A$ that includes the objects from groups $G_1$, $G_2$, and $G_4$ will return $\tau+f-1$ records.
If $t < \tau+f$ is enough to produce an evidence of an effective read, then the auditor will report this read from reader $p_{r2}$ (i.e., not an effective read), violating the strong accuracy property.
\end{proof}

\begin{theorem}{It is impossible to satisfy both \textbf{completeness} and \textbf{strong accuracy} of auditable registers.}
\label{the:strongcompl}
\end{theorem}

\begin{proof}
Consider an auditable register implemented using the same five subset groups of objects $G_{1-5}$ from Lemma~\ref{lem:strong} (depicted in Figure~\ref{fig:scenario2}).
Without loss of generality, let us assume $\tau \geq 2f+1$ is required to guarantee the completeness of audit operations (Lemma~\ref{lem:completeness}).
As a consequence, $G_2$ contains $\tau -2f \geq 1$ object and $G_5$ contains $n-\tau-2f+1 \leq n - 4f$ objects.

A reader $p_{r1}$ executes a read operation in groups $G_{1-3}$.
Faulty objects from $G_1$ return their blocks for value~$v$ (but do not log the read operation).
Reader $p_{r1}$ receives exactly~$\tau$ correct data blocks and obtains the original value~$v$ (i.e., an effective read).

A writer $p_{w}$ leaves incomplete an operation to write a value $x$, similarly to the configuration of Lemma~\ref{lem:strong}, which has been received by objects in all groups but $G_2$ and $G_4$.
A second reader $p_{r2}$ executes a read operation that receives $\tau-1$ blocks for value~$v$ from a providing set $P_{p_{r2},v}$.
These $\tau-1$ blocks are insufficient to recover the original value~$v$ (i.e., it does not represent an effective read).
Faulty objects from $G_1$ decide to mimic groups $G_2$ and $G_4$ and log the read of $v$ by $p_{r2}$.

After these two reads, any auditing quorum~$A$ that includes $G_4$ receives the same number or more records 
of the read from $p_{r2}$ (i.e., not an effective read) than records
of the effective read from $p_{r1}$.
As a consequence, it is impossible to define a single reporting threshold~$t$ to be used to produce evidences without violating either completeness or strong accuracy.

For instance, the auditing quorum $A = G_{1} \cup G_{2} \cup G_{4} \cup G_{5}$ receives $\tau+f-1$ records $\langle p_{r2}, \emph{label}(b_{v_k}) \rangle_k$ and $\tau-2f$ records $\langle p_{r1}, \emph{label}(b_{v_k}) \rangle_k$.
Defining $t \leq \tau-2f$ is enough to report the effective read from $p_{r1}$, but it violates the strong accuracy by also reporting the (not effective) read from $p_{r2}$.
Alternatively, defining $t \geq \tau+f$ satisfies the strong accuracy because it does not report the read from $p_{r2}$, but it violates the completeness because it does not report the effective read from $p_{r1}$.
\end{proof}

\textbf{Remark:} Adding any number of objects to the scenario of Figure~\ref{fig:scenario2} does not change the impossibility result of Theorem~\ref{the:strongcompl}.
The reason is that with the unlimited concurrency in our model, each additional object may have a value different from all values stored in the other objects.

\section{Audit Algorithm}
\label{sec:algs}

We present a generic auditability algorithm for register emulations and prove that all bounds from the previous section (Lemmata~\ref{lem:completeness}--\ref{lem:strong} and Theorem~\ref{the:weakcompl}) are tight in our system model when using this algorithm.
Concretely, we prove that this algorithm satisfies the completeness property with $t \geq 1$ and $\tau \geq 2f+1$ (Lemma~\ref{lem:alg-completeness}) and the weak accuracy with $t \geq f+1$ (Lemma~\ref{lem:alg-weak}).
Then, we prove that it supports both completeness and weak accuracy with $t \geq f+1$ and $\tau \geq 3f+1$ (Theorem~\ref{the:alg-complweak}).
Finally, we prove it supports the strong accuracy property alone with $t \geq \tau+f$ (Lemma~\ref{lem:alg-strong}).

The implementation for \emph{a-audit()} is presented in Algorithm~\ref{alg:audit}.
It starts with an empty set $E_A$ that will be used to store the evidences attesting the effective reads and an array of empty logs to store the logs it will receive from objects (Line~2).
It then queries $n$~storage objects to obtain the list of records on objects' logs (i.e., \emph{$o_k$.getLog()}), waits the response from at least $n-f$ of them, and stores these responses in the array of logs $L$ (Lines~3--4).

For each previously seen record (Line~5), it adds the identifier of every object containing this record to an evidence $\mathcal{E}_{p_r,v}$ (Line~6).
If this evidence contains more than~$t$ identifiers, then it refers to an effective read and is added to the reporting set of evidences $E_A$ (Line~7).
After verifying all records, the audit operation returns the set $E_A$ (Line~8), which is used by auditors to report that the detected readers have effectively read the mentioned data values.

\begin{algorithm}[t!]
    \begin{algorithmic}[1]
		\State \textbf{function} \emph{a-audit}( )
		\State \hskip2em $E_A \leftarrow \varnothing$, $L[1..n] \leftarrow \varnothing$
		\State \hskip2em \textbf{parallel for} {$1~\leq~k~\leq~\mathit{n}$} \textbf{do} $L[k] \leftarrow o_k.getLog()$
		\State \hskip2em \textbf{wait} {$|\{ k : L[k] \neq \varnothing \} | \geq n-f $}
		\State \hskip2em \textbf{for all} {$ \langle  p_r, \emph{label}(b_{v}) \rangle \in \bigcup_{k \in \{1..n\}} L[k] $} \textbf{do}
		\State \hskip4em \textbf{for} {$1~\leq~k~\leq~\mathit{n}$} \textbf{do} \textbf{if} $\langle p_r, \emph{label}(b_{v}) \rangle \in L[k]$ \textbf{then} $\mathcal{E}_{p_r,v} \leftarrow \mathcal{E}_{p_r,v} \cup \{k\} $
		\State \hskip4em \textbf{if} {$|\mathcal{E}_{p_r,v}|\geq t$} \textbf{then} $ E_A \leftarrow E_A \cup \{ \mathcal{E}_{p_r,v} \}$
		\State \hskip2em \textbf{return} $E_A$
	\end{algorithmic}
	\caption{\label{alg:audit}The \emph{a-audit()} algorithm.}
\end{algorithm}

\begin{lemma}{}
\label{lem:alg-completeness}
Algorithm~\ref{alg:audit} satisfies the \textbf{completeness} property of auditable registers with $t \geq 1$ and $\tau \geq 2f+1$.
\end{lemma}

\begin{proof}
Based on Lemma~\ref{lem:min-records}, if $\tau \geq 2f+1$, any auditing quorum~$A$ receives, for every effective read, at least one record from a correct storage object that has also participated on the providing set of this read (i.e., $\tau -2f \geq 1$).
Consequently, Algorithm~\ref{alg:audit} satisfies completeness by obtaining some record from any effective read.
\end{proof}

\begin{lemma}{Algorithm~\ref{alg:audit} satisfies the \textbf{weak accuracy} of auditable registers with $t \geq f + 1$.}
\label{lem:alg-weak}
\end{lemma}

\begin{proof}
To support the weak accuracy of auditable registers, Algorithm~\ref{alg:audit} simply needs to make the~$f$~records from faulty objects insufficient to create an evidence reporting an effective read from a reader $p_r$.
Therefore, $t \geq f+1$ ensures that any auditing quorum~$A$ includes at least one correct object that received a read request from this reader and participates in the providing set $P_{p_r}, v$.
\end{proof}

\begin{theorem}{Algorithm~\ref{alg:audit} satisfies both \textbf{completeness} and \textbf{weak accuracy} of auditable registers with $\tau \geq 3f+1$.}
\label{the:alg-complweak}
\end{theorem}

\begin{proof}

Assuming $\tau \geq 3f+1$ directly satisfies completeness (Lemma~\ref{lem:alg-completeness}).
Based on Lemma~\ref{lem:min-records}, any auditing quorum~$A$ receives at least $\tau - 2f \geq f+1$ records from correct storage objects that have participated on the providing sets of each effective read.
As proved in Lemma~\ref{lem:alg-weak}, $t \geq f+1$ is enough for Algorithm~\ref{alg:audit} to satisfy also weak accuracy.
\end{proof}

A practical consequence of Theorem~\ref{the:alg-complweak} is that existing information dispersal schemes that support fast reads using $n \geq \tau + 2f$ objects, such as Hendricks \emph{et.~al}~\cite{hendricks2007low}, would require at least $n~\geq~5f+1$ to support this weak auditability.

\begin{lemma}{Algorithm~\ref{alg:audit} satisfies the \textbf{strong accuracy} of auditable registers with $t \geq \tau + f$.}
\label{lem:alg-strong}
\end{lemma}
\begin{proof}
To support the strong accuracy property of auditable registers, evidences must be created using at least~$\tau$ records attesting the read of the same value from correct storage objects that have participated in the providing sets of each effective read.
Since $f$~faulty objects can mimic these objects and also participate in an auditing quorum~$A$, this number~$f$ must be added to the threshold number~$t$ of records required to create an evidence.
These two requirements make strong accuracy satisfiable with $t \geq \tau +f$ by accessing any auditing quorum~$A$ because, in this case, the records from the $f$~faulty objects make no difference when reporting a value effectively read by a reader.
\end{proof}

\section{Alternative Models for the Algorithm}
\label{sec:alternative}

As proved in Theorems~\ref{the:weakcompl} and~\ref{the:strongcompl}, respectively, weak auditability is impossible with $\tau \leq 3f$ and strong auditability is impossible in our system model.
In the following, we present three modifications to our system model that allow Algorithm~\ref{alg:audit} to overcome these negative results.
More specifically, signing read requests makes weak auditability easier (Section~\ref{subsec:signed-reads}), while totally ordering operations (Section~\ref{subsec:total-order}) or using non-fast reads (Section~\ref{subsec:nonfast}) enables strong auditability.
Table~\ref{tab:bounds} presents an overview of the results in our system model and models with these three modifications.

\begin{table}[!t]
\begin{center}
\begin{tabularx}{\textwidth}{ | X | l | l | l | l | l |}
\hline
 & & {\scriptsize\textbf{Weak}} & {\scriptsize\textbf{Completeness +}} & {\scriptsize\textbf{Strong}} & {\scriptsize\textbf{Completeness +}} \\
{\scriptsize\textbf{Model}} 		& {\scriptsize\textbf{Completeness}} & {\scriptsize\textbf{Accuracy}} & {\scriptsize\textbf{Weak Accuracy}} & {\scriptsize\textbf{Accuracy}} & {\scriptsize\textbf{+ Strong Accuracy}} \\  \hline
{\scriptsize\textbf{Our System Model}} & \multirow{6}{*}{\boldmath$\tau \geq 2f+1$} & $t \geq f+1$ & $\tau \geq 3f+1$ & \multirow{2}{*}{$t \geq \tau + f$} & \multirow{2}{*}{ Impossible} \\ \cline{1-1}\cline{3-4}
{\scriptsize\textbf{Signed Reads (SR)}} &  & \boldmath$t \geq 1$ & \boldmath$\tau \geq 2f+1$ &  &  \\ \cline{1-1}\cline{3-6}
{\scriptsize\textbf{Total Order (TO)}} &  & \multirow{2}{*}{$t \geq f+1$} & \multirow{2}{*}{$\tau \geq 3f+1$} & \multirow{3}{*}{$t \geq f+1$} & \multirow{3}{*}{$\tau \geq 3f+1$} \\ \cline{1-1}
{\scriptsize\textbf{Non-Fast Reads (NF)}} & &  &  & &  \\ \cline{1-1}\cline{3-4}
{\scriptsize\textbf{TO+SR}} & & \multirow{2}{*}{\boldmath$t \geq 1$} & \multirow{2}{*}{\boldmath$\tau \geq 2f +1$} & &  \\ \cline{1-1}\cline{5-6}
{\scriptsize\textbf{NF+SR}} &  & & & \boldmath$t \geq 1$ & \boldmath$\tau \geq 2f+1$ \\ \hline
\end{tabularx}
\end{center}
\caption{Required number~$t$ of collected records for creating read evidences and number~$\tau$ of blocks to read a value in each considered model for each audit property or combination of properties.}
\label{tab:bounds}
\end{table}

\subsection{Signed Read Requests}
\label{subsec:signed-reads}

Digital signatures are tamper-proof mechanisms that enable attesting the authenticity and integrity of received messages.
Correct readers may sign their read requests to mitigate being incriminated by faulty objects.

In general, readers signing requests prevents a faulty object from creating correct records about them if this object has never received a signed request from them before.
It does not modify the number of records from an effective read available to auditors (Lemma~\ref{lem:min-records}) nor does it modify the lower bound on the completeness (Lemma~\ref{lem:alg-completeness}).
However, it directly relates to the weak accuracy property since auditors receiving any correct record with a signed read request is enough for them to attest that a reader has definitely tried to read data in the storage system.

\begin{lemma}{Correct readers signing read requests in our model allows Algorithm~\ref{alg:audit} to satisfy \textbf{weak accuracy} with $t~\geq~1$.}
\label{lem:signedread-weak}
\end{lemma}
\begin{proof}
If read requests from correct readers are signed and these readers have never sent signed requests to the storage system, faulty objects will never create correct records about these readers.
As a result, $t \geq 1$ ensures that auditors will only create evidences for the readers that have tried to read data by sending signed requests to the storage objects.
\end{proof}

\begin{theorem}{Correct readers signing read requests in our model allows Algorithm~\ref{alg:audit} to satisfy both \textbf{completeness} and \textbf{weak accuracy} with $\tau \geq 2f+1$.}
\label{the:signedread-compl-weak}
\end{theorem}

\begin{proof}
With $\tau \geq 2f+1$ and based on Lemma~\ref{lem:min-records}, any auditing quorum~$A$ receives at least $\tau -2f = 1$ record from a correct storage object that has participated on each effective read.
As proved in Lemma~\ref{lem:signedread-weak}, $t \geq 1$ is enough to satisfy weak accuracy in our model when read requests are signed.
Moreover, this limitation does not change the lower bound of Lemma~\ref{lem:alg-completeness}, which means that $\tau \geq 2f+1$ is enough for Algorithm~\ref{alg:audit} to also satisfy completeness in our model.
\end{proof}

Faulty objects can still repurpose generic signed read requests to create correct records of an arbitrary value.
For that reason, signing generic read requests does not modify the lower bounds of strong accuracy (Lemma~\ref{lem:alg-strong}) and supporting it in conjunction with completeness remains impossible (Theorem~\ref{the:strongcompl}).

Alternatively, correct readers can sign the read request with the label of the exact value they intend to read (e.g., a timestamp).
However, additional modifications are required to enable readers learning this label before sending the read request, as will be explained in Section~\ref{subsec:nonfast}.

\subsection{Total Order}
\label{subsec:total-order}

Serialising operations using total order broadcast~\cite{hadzilacos1994modular} allows our system to execute operations sequentially.
By doing so, we limit the number of different values in storage objects to only one value since high-level \emph{a-read} and \emph{a-write} operations are executed in total order.
In the worst case, the system will have $f$~faulty objects with incorrect read records.
With this limitation, strong accuracy becomes easier and can be satisfied together with completeness.

\clearpage
\begin{lemma}{Totally ordering operations in our model allows Algorithm~\ref{alg:audit} to satisfy \textbf{strong accuracy} with $t \geq f + 1$.}
\label{lem:strong-sequential}
\end{lemma}

\begin{proof}
If operations are totally ordered in the system, all objects will store blocks of the same value~$v$, and no object returns blocks of other values.
No read operation will result in correct records for more than one data value, and every read operation results in the creation of records in at least $n-f$ objects.
The worst case is when the $f$~faulty objects log incorrect read records for a value other than~$v$.
As a result, any auditing quorum~$A$ that includes the faulty objects will receive $f$~records attesting the read of a value other than $v$.
Requiring $t \geq f +1$ to create an evidence guarantees that at least one correct object has also participated in the providing set $P_{p_{r},v}$, which guarantees that $n-f$ objects will also return their blocks, allowing $p_r$ to effectively read the value~$v$.
\end{proof}

\begin{theorem}{
Totally ordering operations in our model allows Algorithm~\ref{alg:audit} to satisfy both \textbf{completeness} and \textbf{strong accuracy} with $\tau \geq 3f+1$.}
\label{the:complstrong-noconcurrency}
\end{theorem}
\begin{proof}
With $\tau \geq 3f+1$ and based on Lemma~\ref{lem:min-records}, any auditing quorum~$A$ receives at least $\tau -2f = f+1$ records from correct storage objects that have participated on each effective read of value~$v$.
As proved in Lemma~\ref{lem:strong-sequential}, $t \geq f +1$ is enough to satisfy strong accuracy in our model when operations are totally ordered.
Moreover, this limitation does not change the lower bound of Lemma~\ref{lem:alg-completeness}, which means that $\tau \geq 3f+1$ is enough for Algorithm~\ref{alg:audit} to also satisfy completeness in our model.
\end{proof}

As proved in Lemma~\ref{lem:signedread-weak}, signing generic read requests reduces the lower bound of weak accuracy, which also benefits the model with total ordering (see Table~\ref{tab:bounds}).
However, as previously mentioned, these generic signatures do not modify the lower bounds of strong accuracy.

\subsection{Non-fast Reads}
\label{subsec:nonfast}

There are read algorithms that use more than one communication round (i.e., a \emph{non-fast read}~\cite{Guerraoui06fastreads}) to ensure correct readers will only fetch the blocks of the most up-to-date value stored in the register.
For instance, \textsf{DepSky-CA}~\cite{DepSky} is a register emulation in which the first round of a read obtains the label of the most up-to-date value available in $n-f$ objects, while the second round actually reads only the coded blocks for that specific value.
As a consequence, correct objects log reads from readers only if they hold the most up-to-date value available in at least $n-f$ objects.

\begin{lemma}{Applying Algorithm~\ref{alg:audit} to \textsf{DepSky-CA} protocol satisfies \textbf{strong accuracy} with $t~\geq~f+1$.}
\label{lem:strong-nonfast}
\end{lemma}
\begin{proof}
Let us assume registers follow \textsf{DepSky-CA} protocol~\cite{DepSky} with two communication rounds in read operations.
No read operation will result in correct records for more than one data value, and every read operation results in the creation of records in at least $n-f$ objects.
In the worst-case scenario, the $f$~faulty storage objects can only create at most $f$~incorrect records for an arbitrary value.
Requiring $t \geq f +1$ to create an evidence and report an effective read guarantees that it is achieved only when the reported value was effectively read from $n-f$ objects.
\end{proof}

\begin{theorem}{
Applying Algorithm~\ref{alg:audit} to \textsf{DepSky-CA} protocol satisfies both \textbf{completeness} and \textbf{strong accuracy} with $\tau \geq 3f+1$.}
\label{the:nonfast-complstrong}
\end{theorem}

\begin{proof}
Let us assume registers follow \textsf{DepSky-CA} protocol~\cite{DepSky} with two communication rounds in read operations.
Assuming $\tau \geq 3f+1$ and based on Lemma~\ref{lem:min-records}, any auditing quorum~$A$ receives at least $\tau - 2f = f+1$ records from correct storage objects that have participated on the providing set of each effective read.
As proved in Lemma~\ref{lem:strong-nonfast}, $t \geq f +1$ is enough to satisfy strong accuracy when using \textsf{DepSky-CA} algorithm.
Moreover, this limitation does not change the lower bound of Lemma~\ref{lem:alg-completeness}, which means that $\tau \geq 3f+1$ is also enough to satisfy completeness in our model.
\end{proof}

As mentioned in Section~\ref{subsec:signed-reads}, readers can sign their read requests for any value (i.e., generic signatures) or for specific values.
Non-fast reads can leverage the latter to allow the system to protect against faulty objects using signed read requests to record reads from other values than the one intended by the reader.
It contributes to reduce the lower bound of strong accuracy and makes easier to support it with completeness.

\begin{lemma}{Applying Algorithm~\ref{alg:audit} to \textsf{DepSky-CA} protocol with specific signed read requests satisfies \textbf{strong accuracy} with $t~\geq~1$.}
\label{lem:strong-nonfast-signedv}
\end{lemma}
\begin{proof}
Let us assume registers follow \textsf{DepSky-CA} protocol~\cite{DepSky} with two communication rounds in read operations.
If readers learn, in the first communication round, the label of the value they can effectively read and use it to sign the read request specifically for that value, faulty objects will never create correct records of these readers and that value. 
Additionally, no read operation will result in correct records for more than one data value, and every read operation results in the creation of correct records in at least $n-f$ objects.
As a result, $t \geq 1$ ensures that auditors will only create evidences for the readers that have effectively read that value by sending specific signed requests to the storage objects.
\end{proof}

\begin{theorem}{
Applying Algorithm~\ref{alg:audit} to \textsf{DepSky-CA} protocol with specific signed read requests satisfies both \textbf{completeness} and \textbf{strong accuracy} with $\tau \geq 2f+1$.}
\label{the:nonfast-signedv-complstrong-noconcurrency}
\end{theorem}

\begin{proof}
Let us assume registers follow \textsf{DepSky-CA} protocol~\cite{DepSky} with two communication rounds in read operations and signed read requests for specific values.
Assuming $\tau \geq 2f+1$ and based on Lemma~\ref{lem:min-records}, any auditing quorum~$A$ receives at least $\tau -2f = 1$ record from correct storage objects that have participated on the providing set of each effective read.
As proved in Lemma~\ref{lem:strong-nonfast-signedv}, $t \geq 1$ is enough to satisfy strong accuracy when using \textsf{DepSky-CA} algorithm with specific signatures.
Moreover, this limitation does not change the lower bound of Lemma~\ref{lem:alg-completeness}, which means that $\tau \geq 2f+1$ is also enough to satisfy completeness in our model.
\end{proof}

A practical consequence of Theorems~\ref{the:nonfast-complstrong} and~\ref{the:nonfast-signedv-complstrong-noconcurrency} is that \textsf{DepSky-CA} protocol~\cite{DepSky} would require at least $n~\geq~5f+1$ (without signed reads) and $n~\geq~4f+1$ (with signed reads for specific values) to support strong auditability.

\section{Final Remarks}

This paper defined the notion of auditable registers, their properties, and established tight bounds and impossibility results for auditable storage emulations in the presence of faulty storage objects. 
Our system model considered read-write registers that securely store data using information dispersal and support fast reads. 
In such a scenario, given a minimum number~$\tau$ of data blocks required to recover a value from information dispersal schemes and a maximum number $f$ of faulty storage objects
(1)~auditability is impossible with $\tau \leq 2f$; (2)~when fast reads are supported, $\tau \geq 3f+1$ is required for implementing weak auditability, while stronger auditability is impossible; (3)~signing read requests overcomes the lower bound of weak auditability; and (4)~totally ordering operations or using non-fast reads can provide such strong auditability.

The system model considered in this paper is strongly grounded on practical systems deployed in multi-cloud environments(e.g.,~\cite{racs,DepSky,resch2011aontrs}).
As future work, it would be interesting to study how our results will change if considering different models, such as server-based (where base objects can run protocol-specific code and communicate with each other) and synchronous systems.

\paragraph{Acknowledgements.}

This work was partially supported by the \textsf{Funda\c{c}\~ao para a Ci\^encia e a Tecnologia}~(\textsf{FCT}, Portugal), through the \textsf{LASIGE} research unit~(UIDB/00408/2020 and UIDP/00408/2020) and the \textsf{IRCoC} project~(PTDC/EEI-SCR/6970/2014), and by the \textsf{European Commission}, through the \textsf{DiSIEM}~(H2020-IA-700692) project.

\bibliographystyle{plain}
\bibliography{main}
\end{document}